\documentclass[11pt]{amsart} 
\usepackage[margin=1 in]{geometry}
\usepackage{array}
\usepackage{amsthm}
\usepackage{amssymb}
\usepackage{amsfonts}
\usepackage{graphicx}
\usepackage{todonotes}
\usepackage{mathrsfs}
\usepackage{amsmath}
\usepackage{graphicx}
\usepackage{hyperref}
\usepackage{setspace}

\newtheorem{theorem}{Theorem}[section]
\newtheorem{cor}{Corollary}
\newtheorem{lemma}[theorem]{Lemma}
\newtheorem{prop}{Proposition}[section]
\newtheorem{defin}{Definition}

\newcommand{\im}{{\normalfont im}}

\title{A lattice model with Fibonacci degree of degeneracy} 
\author{Athena Wang}

\begin{document}
\setstretch{1.5}
\maketitle

\begin{center}
    Abstract
\end{center}
In this paper, we explore two different methods of finding the degrees of degeneracy for lattice model systems, specifically constructing one with a Fibonacci degree of degeneracy. We also calculate the number of ground states per site as the golden ratio $(\phi)$ for the system that we constructed and extend our results to systems with $k-$Step Fibonacci degrees of degeneracy. Finally, I end with a few open questions that we may examine for future works.

\section{Introduction}

Lattice models, as defined within mathematical physics, provide a discrete, grid-like structure for depicting physical systems. Originating from condensed matter physics to model crystalline structures, these models have since pervaded numerous areas of theoretical physics due to their unique properties. The investigations of these models offer invaluable insights into fundamental phenomena, including phase transitions, magnetization, and scaling behavior, and contribute significantly to the comprehension of quantum field theory \cite{Baxter1982}. Moreover, they provide a practical means of approximating continuum theories, effectively introducing an ultraviolet cutoff to prevent divergences and facilitating numerical computations \cite{luscher1988there}.

In classical physics, lattice models are generally described by an energy function on the phase space, with the Ising model \cite{Selinger_2016} being one of the prototypical examples in this context. Quantum mechanics, meanwhile, with its peculiar characteristics of superposition and entanglement, offers a distinct approach to using lattice models. To mathematically encapsulate these unusual features, we must employ more advanced mathematical structures, such as Hilbert spaces and Hermitian operators, fundamental concepts that form the baseline of quantum theory. These more advanced structures also make it possible to describe richer phenomena, with ground-state degeneracy being an interesting example.

In a quantum lattice model, we have a finite-dimensional Hilbert space $V$ representing the spin on each site and a local interaction that involves only a finite number of nearby sites. The Hilbert space of the whole system is a tensor product of all the $V$s labeled by different sites (represented by $V^{\otimes n}$), and the Hamiltonian is the sum of a number of summands that apply a local transformation. An eigenvector with the lowest eigenvalue is called a ground state---or the state using the least amount of energy---and the dimension of the corresponding eigenspace is called the number of ground states, or degree of degeneracy. In this paper, we investigate an example such that the degree of degeneracy is the Fibonacci number $F_{n+1}$ (where $n$ is the size of the system). In particular, the average number of ground states per site is $\lim \sqrt[n]{F_{n+1}}=\frac{1+\sqrt{5}}{2}$.

More precisely, let us consider a lattice model on a line with length $n$, with each number $i=1,2,\cdots,n$ labeling a \textit{site} on which we have a \textit{local Hilbert space of states} $V_i\cong V=\langle|0>,|1>\rangle$, or the Hilbert space labeled by the bases (or qubits) $|0>$ and $|1>$. The \textit{space of states for the whole system} is $\otimes_i V_i\simeq V^{\otimes n}$. We will consider a local interaction for each pair of neighboring sites $i,i+1$ given by a Hermitian operator $H_{i,i+1}$ on $V_i\otimes V_{i+1}\simeq V^{\otimes 2}$, which maps $|11>$ to itself and all other standard basis vectors, $|00>,|01>,|10>$, to $0$. The Hamiltonian $H$ of the whole system is an operator on $V^{\otimes n}$ given by the sum of all local interactions $H_{i,i+1}$. In this paper, we specifically investigate the following property of this Hamiltonian $H$:

\begin{theorem}
     $H$ is a non-negative definite Hermitian operator, and the degree of degeneracy of the system (i.e. $\dim(\ker H$)) with Hermitian $H$ acting on $n$ sites is the Fibonacci number $F_{n+1}$.
\end{theorem}
To start, I will split the proof of \textbf{Theorem 1.1} into two parts: the first half, which proves that $H$ is a non-negative definite Hermitian operator and introduces the necessary structures, and the second half, which proves that the degree of degeneracy of the system is $F_{n+1}$.

\section{Hermitians and Hamiltonians}
In this section, I will focus on introducing the Hermitian and the Hamiltonian.

\begin{defin}
    A Hermitian matrix $A$ of a Hermitian operator $H$ has the property that $A^*=A$, where ``$A^*$" represents the conjugate transpose ($\overline{A^T}$) of $A$. Therefore, for every entry $a_{ij}$, representing the entry in row $i$ and column $j$ of matrix $A$, $a_{ij}= \overline{a_{ji}}$, making it an extension of symmetric matrices in the complex numbers.
\end{defin}

Hermitians, with their complex number entries describing different states and electron configurations of a transformation, are an integral part of quantum theory. Just like symmetric matrices, Hermitians simplify the parameterization of a transformation of a system over time.\\

With this definition, I can now check that $H_{i,i+1}$, which maps $|00>\mapsto 0$, $|01>\mapsto 0$, $|10> \mapsto 0$, and $|11>\mapsto|11>$, is a Hermitian.

Our Hermitian matrix $A$ would be
\vspace{3pt}

\begin{center}
    $\begin{bmatrix}
        0 & 0 & 0 & 0\\
        0 & 0 & 0 & 0\\
        0 & 0 & 0 & 0\\
        0 & 0 & 0 & 1\\
    \end{bmatrix}$
\end{center}
\vspace{3pt}
with the four entries on the diagonal ($[0, 0, 0, 1]$) representing exactly how the bases in the system interact with each other when the operator is applied to adjacent sites. Since $A^*=A$, operator $H_{i,i+1}$ is clearly Hermitian.
\begin{defin}
    The Hamiltonian $H=\sum_{i=1}^{n-1} H_{i,i+1}$, the sum of all local interactions $H_{i,i+1}$.
\end{defin}

For the specific Hamiltonian $H$ of our system, I can prove the following results.
\begin{prop}
    The matrix of the Hamiltonian is diagonal.
\end{prop}
\textit{Remark. For this proof, I will be using a 1-1 correspondence between a basis vector with $n$ sites and a standard matrix with $2^n$ entries. For every qubit of the basis, there are two corresponding matrices. As an example,
\vspace{3pt}
$|000...00> \mapsto \begin{bmatrix}
    1\\
    0\\
    0\\
    ...\\
    0\\
\end{bmatrix}$
and $|000...01> \mapsto \begin{bmatrix}
    0\\
    1\\
    0\\
    ...\\
    0\\
\end{bmatrix}$
\vspace{3pt}
There are exactly $2^n$ total basis vectors with $n$ sites and $2^n$ standard matrices with $2^n$ entries.}
\begin{proof}
For a diagonal $2^n\times 2^n$ matrix $A$ with eigenvalues $a_1, a_2, a_3...a_{2^n}$, $Ae_i = a_ie_i$ for standard unit vectors $e_i$, represented by a 1 in the $ith$ row and $0$ in every other row of a column matrix.

Then, since the basis vector mappings of $V\otimes V$ are defined to be multiples of the vectors themselves (with $|00>\mapsto 0, |01>\mapsto 0, |10>\mapsto 0,$ and $|11>\mapsto |11>$),

\begin{center}
    $H_{i,i+1}e_j = \lambda_je_j$ for any $i=1, 2,...n-1$ and $j=1, 2,...2^n$. 
\end{center}

In fact, for an $e_j$ that corresponds to an $n$-site basis vector with the sites $i$ and $i+1$ represented by a 1, $\lambda_j = 1$ (since $|11>\mapsto|11>$). Otherwise, $\lambda_j=0$ (since all other bases map to $0$).

Therefore, the matrix $H_{i,i+1}$ would be diagonal with real entries (either 0 or 1). Since the sum of diagonal matrices will also be diagonal, matrix $H = H_{1,2} + H_{2,3}...+H_{n-1,n}$ will also be diagonal with real entries.
\end{proof}
As diagonal matrices are symmetric, it is then apparent that the following corollaries are true.
\begin{cor}
    $H$ is symmetric.
\end{cor}

\begin{cor}
    $H$ is a non-negative definite system.
\end{cor}

\textit{Remark. ``Non-negative definite systems" are often called ``positive semi-definite systems" in some texts.}

\begin{proof}
From \textbf{Proposition 2.1}, I proved that every $\lambda_j$, or eigenvalue of $H_{i,i+1}$, is either 0 or 1 in our system with Hamiltonian $H$. Since both 0 and 1 are non-negative and the sum of non-negative integers are also non-negative, all the entries on the diagonal---or all the eigenvalues of $H$---will be non-negative. Therefore, it will be a non-negative definite system.
\end{proof}

For any non-negative definite system, we call it a \textit{well-behaved quantum system}.

Now that we understand that Hamiltonian $H$ is a non-negative definite system, I only need to check that $H$ is Hermitian.

\begin{cor}
    $H$ is Hermitian.
\end{cor}
\begin{proof}
    Note that if $A$ was Hermitian and its entry $a_{ij}\in \mathbb{R}$, $a_{ij}=\overline{a_{ji}}=a_{ji}$. Hence if all entries of $A\in\mathbb{R}$, $A$ would be a symmetric matrix, which is also Hermitian. Since $H$ is symmetric (\textbf{Corollary 1}) with real entries, system $H$ would also be Hermitian.
\end{proof}

As a result, the first half of the theorem is proved.

Hamiltonian $H=\sum_{i=1}^{n-1}H_{i,i+1}$ where $H_{i,i+1}$ are Hermitian operators that map $|11>$ to itself and all other standard basis vectors to 0 is a non-negative definite Hermitian operator.

\section{Degree of Degeneracy of $H$}
We will now examine the interesting results given by analyzing the degree of degeneracy of the system $H$.

In order to investigate this problem, I will present two different approaches: one using modular arithmetic and matrices and the other using recursion and qubit basis states. By displaying both, I offer a glimpse into the surprisingly diverse connections between quantum computing and other fields of math.

Let us first define the tensor product, a fundamental concept in the following two approaches.

\begin{defin}
    A tensor product between two matrices is a combination of two states interacting with each other. It is represented by the symbol, $\otimes$.
\end{defin}

\begin{defin}
    A ground state is the lowest eigenspace of a system.
\end{defin}

\begin{defin}
    The degree of degeneracy of a system $H$ is the number of ground states of the system. In this case, it will be the same as $\dim(\ker H)$.
\end{defin}

\subsection{Modular Arithmetic and Matrices.}

I will first define the tensor product in terms of matrices.

\begin{defin}
    The tensor product displayed in matrices is also sometimes called the Kronecker product \cite{Broxson2006-pe}, which is characterized by 
    \begin{center}
    $\begin{bmatrix}
    a_{1,1} & a_{1,2} & ... & a_{1,n}\\
    a_{2,1} & a_{2,2} & ... & a_{2,n}\\
    ... & ... & ... & ...\\
    a_{m,1} & a_{m,2} & ... & a_{m,n}\\
\end{bmatrix} \otimes B =
\begin{bmatrix}
    a_{1,1}B & a_{1,2}B & ... & a_{1,n}B\\
    a_{2,1}B & a_{2,2}B & ... & a_{2,n}B\\
    ... & ... & ... & ...\\
    a_{m,1}B & a_{m,2}B & ... & a_{m,n}B\\
\end{bmatrix}$
\end{center}

where $B$ is another matrix.
\end{defin}

\textit{Remark. If $A$ is an $m \times n$ matrix and $B$ is a $j \times k$ matrix, $A\otimes B$ is a $mj \times nk$ matrix.}\\

Before, we defined that Hamiltonian $H=\sum_{i=1}^{n-1} H_{i,i+1}$. Therefore, $Hv=\sum_{i=1}^{n-1} H_{ii,i+1}v$ where vector $v\in V^{\otimes n}$, and im($H$) = $\sum_{i=1}^{n-1}$ im$(H_{i,i+1})$. Thus, let's take a look at $H_{i,i+1}$ individually.

Because $H_{i,i+1}$ only acts on adjacent sites $i, i+1$, im$(H_{i,i+1})$ is spanned by
\begin{center}
    $<v_1 \otimes v_2 \otimes v_3..v_{i-1}\otimes A(v_i \otimes v_{i+1}) \otimes v_{i+2}...\otimes v_n>$.
\end{center} where $A$ is the matrix of the Hermitian operator.

Note that every vector (or state) $v_i$ can be represented as
$\begin{bmatrix}
    a_0\\
    a_1\\
\end{bmatrix}$
with the $a_0$ representing the probability of the electron being in basis state $|0>$ and $a_1$ representing the probability of the electron being in basis state $|1>$.

Because all $v_i$ are $2\times 1$ matrices, each \im($H_{ii+1})$ can be represented as a $2^n \times 1$ column matrix, where each row corresponds with a basis element of the system.

For example,
\begin{center}
$\begin{bmatrix}
    a_0\\
    a_1\\
\end{bmatrix} \otimes \begin{bmatrix}
    b_0\\
    b_1\\
\end{bmatrix} = \begin{bmatrix}
    a_0b_0\\
    a_0b_1\\
    a_1b_0\\
    a_1b_1\\
\end{bmatrix}$
\end{center}
where $a_0$ and $b_0$ correspond with the basis element $|0>$, $a_1$ and $b_1$ correspond with the basis element $|1>$, $a_0b_0$ corresponds with $|00>$, $a_0b_1$ corresponds with $|01>$, $a_1b_0$ corresponds with $|10>$, and $a_1b_1$ corresponds with $|11>$.

In system $H$, finding the degree of degeneracy is the same as finding $\dim(\ker H)$, as the system is non-negative definite (which means the total energy of the system cannot go below 0) and I can find a state such that the total energy of the system is 0.

I can now use a simple, well-known result in linear algebra that proves that $ker(H)=$\im$(H)^\bot$ \cite{Winters2023}.
\begin{theorem}
    For any matrix $A: \mathbb{R}^n\mapsto \mathbb{R}^m$, \im$(A)^\bot=\ker (A^T)$.
\end{theorem}
\begin{proof}
    Suppose that \im($A)=$ span$\{v_1, v_2,...v_n\}$. Then, $A=\begin{bmatrix}
        \uparrow & \uparrow & ... & \uparrow\\
        v_1 & v_2 & ... & v_n\\
        \downarrow & \downarrow & ... & \downarrow\\
    \end{bmatrix}$.
    
    If vector $x\in$ \im$(A)^\bot$,
    
    \begin{center}
    $\begin{Bmatrix}
        v_1\cdot x=0\\
        v_2\cdot x=0\\
        ...\\
        v_n\cdot x=0\\
    \end{Bmatrix} \Leftrightarrow \begin{bmatrix}
        \leftarrow & v_1 & \rightarrow\\
        \leftarrow & v_2 & \rightarrow\\
        ... & ... & ...\\
        \leftarrow & v_n & \rightarrow\\
    \end{bmatrix}$
    $\begin{bmatrix}
        \uparrow\\
        x\\
        \downarrow\\
    \end{bmatrix}=A^Tx=0$
    \end{center}

    which means that $x\in\ker(A^T)$. The reverse argument can be made for the converse. Therefore, \im$(A)^\bot=\ker(A^T)$.
\end{proof}

\begin{cor}
    If matrix $A$ is symmetric, \im$(A)^\bot=\ker(A)$.
\end{cor}
\begin{proof}
Because $A$ is symmetric, $A^T=A$. Hence, \im$(A)^\bot=\ker(A)$.
\end{proof}
Since $H$ is symmetric (\textbf{Corollary 1}), $\ker(H)$=\im$(H)^\bot$ (\textbf{Corollary 4}).

Therefore, to find the dimension of the kernel of the system, I can find the number of orthonormal basis vectors $x$ such that \im$(H)x=0$. Note that all such vectors are standard column matrices, with one row having a value of 1 and the rest having a value of 0.

If \im$(H)$ is spanned by $\begin{bmatrix}
    r_0\\
    r_1\\
    ...\\
    r_{2^n}\\
\end{bmatrix}$, where every $r_j$ for $j=1,2,...2^n$ is non-negative, each vector $x\in$ \im($H)^\bot$ is a standard basis vector with 1 in row $r_j$ if $r_j=0$ in \im($H)$.

Without ideas regarding dim(\im$(H))$, I can start with smaller cases. Let's denote rows that are not 0 as ``-''. Note that $v$ is a vector, so I will denote its values corresponding to bases $|0>$ and $|1>$ as ``-" too. Thus, $v^{\otimes n}$ is a $2^n\times 1$ column matrix where every entry is ``-".
\begin{center}
Let $n=2$. Then, \im$(H) = A(v \otimes v) = \begin{bmatrix}
    0 & 0 & 0 & 0\\
    0 & 0 & 0 & 0\\
    0 & 0 & 0 & 0\\
    0 & 0 & 0 & 1\\
\end{bmatrix}\begin{bmatrix}
    -\\
    -\\
    -\\
    -\\
\end{bmatrix}=\begin{bmatrix}
    0\\
    0\\
    0\\
    -\\
\end{bmatrix}$ 
\end{center}
So, the basis for the kernel of this matrix will be
\begin{center}
    $\begin{bmatrix}
    1\\
    0\\
    0\\
    0
\end{bmatrix}, \begin{bmatrix}
    0\\
    1\\
    0\\
    0
\end{bmatrix}$, and $\begin{bmatrix}
    0\\
    0\\
    1\\
    0
\end{bmatrix}$.
\end{center}

This kernel has a dimension of 3, exactly the same as the number of 0s in the matrix.
\begin{center}
For $n=3$, \im$(H_{1,2}) = A(v \otimes v) \otimes v = \begin{bmatrix}
0 \\
0 \\
0 \\
- \\
\end{bmatrix} \otimes \begin{bmatrix}
    -\\
    - \\
\end{bmatrix} = \begin{bmatrix}
    0\\
    0\\
    0\\
    0\\
    0\\
    0\\
    -\\
    -\\
\end{bmatrix}$

\im$(H_{2,3}) = v \otimes A(v \otimes v) = $$\begin{bmatrix}
    -\\
    -\\
\end{bmatrix} \otimes \begin{bmatrix}
    0\\
    0\\
    0\\
    -\\
\end{bmatrix} = \begin{bmatrix}
    0\\
    0\\
    0\\
    -\\
    0\\
    0\\
    0\\
    -\\
\end{bmatrix}$

\im($H)=$\im$(H_{1,2})+$\im($H_{2,3}) = \begin{bmatrix}
    0\\
    0\\
    0\\
    -\\
    0\\
    0\\
    -\\
    -\\
\end{bmatrix}$.
\end{center}
Thus, since there are 5 rows of 0 in the matrix, $\dim (\ker(H))=5$. Indeed, the basis for the kernel of the matrix is
\begin{center}
    $\begin{bmatrix}
    1\\
    0\\
    0\\
    0\\
    0\\
    0\\
    0\\
    0\\
\end{bmatrix}, \begin{bmatrix}
    0\\
    1\\
    0\\
    0\\
    0\\
    0\\
    0\\
    0\\
\end{bmatrix}$, $\begin{bmatrix}
    0\\
    0\\
    1\\
    0\\
    0\\
    0\\
    0\\
    0\\
\end{bmatrix}$, $\begin{bmatrix}
    0\\
    0\\
    0\\
    0\\
    1\\
    0\\
    0\\
    0\\
\end{bmatrix}$, and $\begin{bmatrix}
    0\\
    0\\
    0\\
    0\\
    0\\
    1\\
    0\\
    0\\
\end{bmatrix}$.
\end{center}

So, in general, \im$(H_{ii+1})$ can be spanned by
\begin{center}
($v^{\otimes (i-1)})\otimes$$\begin{bmatrix}
    0\\
    0\\
    0\\
    -\\
\end{bmatrix}$ (or $A(v\otimes v$)) $\otimes$ ($v^{\otimes (n-i-1)}$).
\end{center}

The result is a $2^n\times 1$ column matrix, constructed by repeating a $2^{n-i+1}\times 1$ column matrix with a value of 0 in the first $\frac{3}{4}$ rows and 1 in the last $\frac{1}{4}$ rows $2^{i-1}$ times.

Since \im$(H)=\sum_{i=1}^{n-1}$ \im$(H_{i,i+1})$, every row with a value of 1 in the matrix spanning \im($H_{i,i+1})$ for some $i=1, 2,...n-1$ would not have a value of 0 in the final matrix spanning im($H$). Thus, \im$(H)$ will be spanned by a column matrix where non-negative values occur at rows $0, -1,...-(x-1) \pmod{4x}$ for all positive integers $x\leq \frac{2^n}{4}$ if the rows of the matrix spanning \im$(H)$ are labeled from 1 through $2^n$. There will always be some $i=1, 2,...n-1$ such that the repeated $2^{n-i+1}$ matrix has a dimension of $4x$ for a fixed $x\leq \frac{2^n}{4}$.

This generalization can also be viewed from another perspective: rather than using complementary counting, I can count the number of 0s in the matrix spanning \im($H$) directly. The 0s in the final matrix must satisfy $1, 2,,,,3x\pmod{4x}$ for all positive integers $x\leq \frac{2^n}{4}$.

Now, I have all the tools to prove the final half of \textbf{Theorem 1.1}.

\begin{prop}
The dimension of the kernel, or the degree of degeneracy, of the system $H$ is $F_{n+1}$ for a given $n$. Here, $F_0 = 1, F_1 = 1, F_2 = 2,$ etc.
\end{prop}
\begin{proof}
We will begin by using induction on $n$, the number of sites of the system $H$.

Base Case: When $n=2$, $x\leq \frac{2^2}{4}=1$. Then, 0 is in rows $1, 2$, or $3 \pmod {4}$. There are 3 zeroes in the matrix spanning \im$(H)$, which means that $\dim(\ker H) = 3$. This supports our inductive hypothesis.

Inductive Step: Suppose that the degree of degeneracy of $H$ when $n=k-1$ is $F_k$ and when $n=k$ is $F_{k+1}$. Then, when $n=k+1$, $x=2^{k-1}$ adds a new modular equation to satisfy, where 0s can only exist on rows that are $1, 2,...3\cdot2^{k-1}\pmod{2^{k+1}}$. Since $3\cdot2^{k-2}<3\cdot 2^{k-1}$, all the rows of 0s that satisfied the conditions of $n=k$ also satisfy the conditions of $n=k+1$. The rest of the possible rows of 0s, labeled from $2^k+1$ to $3\cdot 2^{k-1}$, if put in $\pmod {2^k},$ will just be $1, 2, 3...2^{k-1}$. Since $2^{k-1} < 3\cdot 2^{k-2}$, the numbers already satisfied the condition required for the case when $n=k$, and I only need to check the conditions for $n=k-1$. Thus, I am checking for the numbered rows that satisfy $1, 2,...3x\pmod{4x}$ when $x\leq \frac{2^{k-1}}{4}$. This is equivalent to checking the case of $n=k-1$. Therefore, since (the number of 0s in the matrix spanning \im($H)$ when $n=k+1)=$ (the number of 0s in the matrix spanning \im($H)$ when $n=k$)+(the number of 0s in the matrix spanning \im($H)$ when $n=k-1$), $F_{n+1}+F_n=F_{n+2}$, which is $\dim(\ker H)$, or the degree of degeneracy of system $H$ when it has $n+1$ sites. This supports our inductive hypothesis.

Hence, our induction is complete.
\end{proof}
By using \textbf{Corollary 3} and \textbf{Proposition 3.1}, I have now proved
\textbf{Theorem 1.1}.

However, there is another interesting result that I can obtain by using \textbf{Proposition 3.1}.

\begin{theorem}
    $F_{n+1}=2^n-(2^{n-2}F_0+2^{n-3}F_1+...2^0F_{n-2})$ when $n\geq 2$.
\end{theorem}
\begin{proof}
    For this proof, I will use complementary counting to count the rows of the final matrix $H$ that the value of 0 cannot be placed in. In total, the sum of the number of rows of the final matrix that can contain 0 and the number of rows of the final matrix that cannot contain 0 will be $2^n$. From \textbf{Proposition 3.1}, $F_{n+1}$ is the number of rows of the final matrix that can contain 0. Therefore, I only need to find the number of rows that cannot contain 0, which is given by finding the rows that fit any one of the modular equations $0,-1...-(x-1)\pmod{4x}$ for some $x\leq \frac{2^n}{4}$. I will also use induction.

    Base Case: When $n=2$, $x\leq \frac{2^2}{4}=1$. So, the only modular equation for the row to satisfy is $0\pmod 4$. There is only 1 row that satisfies this---row 4. Since $1+F_{2+1}=1+3=2^2$, I can re-arrange this to get $F_3=2^2-1$, which supports the inductive hypothesis.

    Inductive Step: Suppose that $F_{n+1}=2^n-(2^{n-2}F_0+2^{n-3}F_1+...2^0F_{n-2})$ for $n=k$, where $F_{k+1}$ is the number of rows that contain 0 in $H$, $2^k$ is the total sum, and $(2^{k-2}F_0+2^{k-3}F_1+...2^0F_{k-2})$ is the number of rows that do not contain 0 in $H$. Then, for $n=k+1$, the number of rows that satisfy the modular equations from case $n=k$ is doubled since the number of rows in the column matrix is doubled. Another modular equation is also added: $0,-1,...-(2^{k-1}-1)\pmod{2^{k+1}}$. However, some of the rows satisfy multiple modular equations, meaning that I only want to count the new rows that add to our non-zero row count. Therefore, out of the rows $0,-1,...-(2^{k-1}-1)\pmod{2^{k+1}}$, I subtract off all the modular conditions that can already be satisfied by the earlier case $n=k$. There are exactly $2^{k-1}$ modular conditions, so every modular condition set by case $n=k$ will only appear at most once in $0,-1,...-(2^{k-1}-1)\pmod{2^{k+1}}$. However, every row above $-(2^{k-1}-1)$ in the case $n=k$ will not be considered in the new condition for case $n=k+1$ (as there are $2^k$ rows for the case $n=k$), so I must subtract the number of rows fitting case $n=k-1$. Note that this is because all rows above $-(2^{k-1}-1)$ are found by fitting the conditions for the case $n=k-1$. Therefore, the number of new rows that add to my non-zero row count is $2^{k-1}-(2^{k-2}F_0+2^{k-3}F_1+...2^0F_{k-2}-(2^{k-3}F_0+2^{k-4}F_1+...2^0F_{k-3}))=2^{k-1}-(2^{k-3}F_0+2^{k-4}F_1+...2^0F_{k-3}+2^0F_{k-2})=F_{k}-F_{k-2}=F_{k-1}$. Then, for case $n=k+1$, the number of rows that are non-zero is $2\cdot(2^{k-2}F_0+2^{k-3}F_1+...2^0F_{k-2})+F_{k-1}=2^{k-1}F_0+2^{k-2}F_1+...2^0F_{k-1}$. Since $F_{k+2}$ is the number of rows that contain 0 and $2^{k+1}$ is the total number of rows in the matrix spanning \im($H$) for system $H$ with $k+1$ sites, our inductive hypothesis is proven.

    Thus, this theorem holds true.
\end{proof}

\subsection{Recursion and Qubit Basis States}

We have now already defined that $H_v = \sum_{i=1}^{n-1} H_{i,i+1}v$ where $v\in V ^{\otimes n}$. However, since every vector $v$ can be expressed as a linear combination of the basis vectors, I only need to consider when $v$ is a basis vector to create the basis of the kernel.

Now, since only 2 adjacent qubits interact at a time with operator $H_{i,i+1}$, $H_{i,i+1}v = v$ if and only if $v = |...11...>$ when the 1s represent the spins of site $i$ and $i+1$, respectively. Otherwise, $H_{i,i+1}v = 0$ if $v=|...00...>$, $|...01...>$, or $|...10...>$ at sites $i$ and $i+1$ since they map it to 0.

So, $Hv= \alpha v$ where $\alpha$ is the number of adjacent 1s in $v$. So, $\ker H = \{v: Hv = 0\}$ and $\alpha = 0$. Therefore, there is a bijection between the number of sequences of length $n$ such that there are no adjacent ones.

This is a typical recursion problem. For an $n$-digit sequence, there are two possibilities for the last digit: $0$ or $1$. If the last digit is 0, I am reduced to finding the number of $n-1$ digit sequences that do not have adjacent ones. If the last digit is 1, the second to last digit must be a 0 to avoid adjacent ones, and I am reduced to finding the number of $n-2$ digit sequences that do not have adjacent ones. Now, I can take a look at the base cases to see the recursive sequence.

When $n=2$, there are 3 sequences that do not have adjacent ones---namely, $|00>, |01>, $ and $|10>$. When $n=3$, there are 5 sequences that do not have adjacent ones---namely, $|000>, |001>, |010>, |100>,$ and $ |101>$. Then, when $n=4$, the number of sequences that do not have adjacent ones is 3+5=8. When $n=5$, it is $5+8=13$. This is the Fibonacci sequence, where the number of $n$-digit sequences that do not have adjacent ones is $F_{n+1}$. Therefore, the dimension of the kernel on a Hilbert space with $n$ sites is $F_{n+1}$ where $F_0 = 1$.

This second approach provides a much more efficient result and can possibly be applied to future research with other systems $H$. 

However, both approaches provide some insight into quantum lattice models and their ground states.

\section{Ground States Per Site}

Let's denote all the ground states per site as $V_0$. Then, all the ground states in $H$ can be represented by $V_0^{\otimes n}$, so there are ($\dim V_0)^n$ total ground states.

Thus, for our problem, $(\dim V_0)^n=F_{n+1}$, the total number of ground states. So, $\dim V_0 = \sqrt[n]{F_{n+1}}$. From Binet's Formula, $F_n=\frac{(\frac{1+\sqrt5}{2})^n-(\frac{1-\sqrt5}{2})^n}{\sqrt5}$. When $n$ is large, $F_n\approx \frac{\phi^n}{\sqrt{5}}$ where $\phi=\frac{1+\sqrt5}{2}$ (see \cite{Clancy} for full proof). Then, as $n$ converges to infinity, \[\sqrt[n]{F_{n+1}} \approx \lim_{n\rightarrow\infty} \sqrt[n]{\frac{\phi^{n+1}}{\sqrt5}}=\frac{\phi^{\frac{n+1}{n}}}{\sqrt5^{\lim_{n\rightarrow\infty}\frac{1}{n}}}=\frac{\phi}{\sqrt5^0}=\phi.\] Formally, $\dim V_0=\frac{1+\sqrt{5}}{2}$, which means that the number of ground states per site, on average, is somewhere between 1 and 2.

We have now seen an interesting example of a quantum lattice model where despite the number of states on each site being an integer $F_{n+1}$, the number of ground states on each site is an irrational number. It's natural to conjecture that for more general systems, the number of ground states on each site---if not exactly 1 or 2---will be positive algebraic integers, all of whose Galois conjugates have strictly smaller absolute values. Note that no system will have exactly 0 ground states, as there must be some state(s) that have the lowest energy.

\section{More Generalized Systems With $k-$Step Fibonacci Degrees of Degeneracy}
In this section, I will construct lattice models that have a $k-$Step Fibonacci degree of degeneracy and explore some of its properties, including looking into the conjecture for more generalized systems.

\begin{prop}
    Similar to before, systems $H$ where only the basis element $|111...1>\mapsto |111...1>$ and everything else maps to 0 will give a $k-$Step Fibonacci degree of degeneracy, where $k$ is the number of sites interacting at once (or the number of digits in each basis element).
\end{prop}
\begin{proof}
Like in our original method, I can use recursion to find the degree of degeneracy of the $n$-site system. The problem is equivalent to finding the number of $n$-length sequences of 0s and 1s without a $k$-length sequence of 1s. Now, I can split the problem into cases. If the last digit is 0, I am reduced to finding the number of $n-1$ length sequences without $k$ consecutive 1s (giving us a term of $S_{n-1}$). If the last digit is a 1, I have several different possibilities for ending sequences, all of which are in the form of 011...1. If the length of the ending sequence is 2 (so the ending sequence is 01), I am reduced to finding the number of $n-2$ length sequences without $k$ consecutive ones (giving us a term of $S_{n-2}$). If the length of the ending sequence is 3 (so the ending sequence is 011), I am reduced to finding the number of $n-3$ length sequences without $k$ consecutive ones (giving us a term of $S_{n-3}$). This will continue until the length of the ending sequence is $k$, which can simplify to the problem of finding the number of $n-k$ length sequences without $k$ consecutive ones (giving us a term of $S_{n-k}$). Therefore, the degree of degeneracy of every system $H$ with only the basis element $|111...1>$ mapping to itself will satisfy $S_n=S_{n-1}+S_{n-2}+...+S_{n-k}$. Let me denote the $m$th term of a $k-$Step Fibonacci sequence as $F_{k,m}$. Notice that when the number of sites is $n=k$, there are $2^k-1$ ground states (every possible combination of $k$ digits except $|111...1>)$. Then, since $F_{k,0}=1, F_{k,1}=1, F_{k,2}=2,...F_{k,k-1}=2^{k-2}$, $F_{k, k}=2^{k-1}$ and $F_{k, k+1}=2^k-1$, which is the degree of degeneracy when there are $n=k$ sites in total. Thus, the degree of degeneracy for such systems $H$ with $k$ sites interacting at once would be $F_{k,n+1}$.
\end{proof}

Now, to prove the convergence of the number of ground states per site, it is much easier to work with explicit formulas of the sequences rather than the recursive form. So, I will use the Characteristic Root Technique \cite{levin2015discrete}.

\begin{defin}
    The characteristic polynomial of a recursion $S_n=a_1S_{n-1}+a_2S_{n-2}+...a_kS_{n-k}$ is $x^k-a_1x^{k-1}-a_2x^{k-2}-...-a_k=0.$
\end{defin}
\begin{lemma}[Characteristic Root Technique]
    If the characteristic roots of a linear recurrence are all distinct, the linear recurrence has an explicit formula that can be expressed as
    \begin{center}
        $S_n=\alpha_1r_1^n+\alpha_2r_2^n+...+\alpha_kr_k^n$
    \end{center}
    where every $\alpha_i$ is a constant coefficient and every $r_i$ is a root of the characteristic polynomial of the recurrence. If there are some roots that are equal, let us suppose that we have $$r_s=r_{s+1}=...=r_t.$$ Then, in the original formula, we can replace the term, $$\alpha_sr_s^n+\alpha_{s+1}r_{s+1}^n+...+\alpha_tr_t^n$$ with $$(\alpha_sn^{t-s}+\alpha_{s+1}n^{t-s-1}+...+\alpha_{t-1}\cdot n+\alpha_t)\cdot r_s^n.$$
\end{lemma}
\begin{proof}
This is already a well-known result. See \cite{AoPS Wiki} for full proof.
\end{proof}

\begin{prop}
For systems $H$ where the degree of degeneracy can expressed as a term in a linear recursion, the number of ground states per site will converge to a number between 1 and 2.
\end{prop}
\begin{proof}
    From \textbf{Lemma 5.1}, every term of a linear recurrence where the characteristic polynomial has distinct roots can be expressed by $a_1r_1^n+a_2r_2^n+...+a_kr_k^n$ where $r_i$ are roots of the characteristic polynomial. Since the characteristic polynomials of all such systems $H$ have integer coefficients, if there is an imaginary number root, the conjugate imaginary number is also a root. Therefore, suppose that I have some root $r_i=a+bi=m_i(\cos(\theta)+i\sin(\theta))$. Then, I will let some other root $r_j$ be $a-bi=m_i(\cos(\theta)-i\sin(\theta))$. So,
    $$a_ir_i^n+a_jr_j^n=a_im_i^n((\cos(n\theta)+i\sin(n\theta))+a_jm_i^n(\cos(-n\theta)+i\sin(-n\theta)).$$
    Therefore, I can rewrite the explicit form of the recurrence as $$a_1r_1^n+...+a_lr_l^n+b_1m_1^n\cos(n\theta_1)+id_1m_1^n\sin(n\theta_1)+...+b_pm_p^n\cos(n\theta_p)+d_pm_p^n\sin(n\theta_p)$$
    where $b_i=a_i+a_j.$
    
    Now, I can first suppose that $|r_1|<|r_2|<...<|r_l|$ and $|m_i|<|r_l|$ for all $i=1, 2, ... p$. Then, $r_1=c_1r_l, r_2=c_2r_l,...r_{l-1}=c_{l-1}r_l$ for all $|c_i|<1.$ I can also get that $m_1=e_1r_l, m_2=e_2r_l, ... m_p=e_pr_l$ for all $|e_i|<1$. Then,
    $$\sqrt[n]{S_n}=\sqrt[n]{(\sum_{i=1}^{l-1} a_ic_i^n + a_l+\cos(n\theta)\sum_{i=1}^{p} b_ie_i^n+i\sin(n\theta)\sum_{i=1}^{p} d_ie_i^n)r_l^n}.$$
    Since $c_i^n, e_i^n\rightarrow 0$ as $n\rightarrow\infty$, 
    $$\lim_{n\rightarrow\infty} \sqrt[n]{S_n}=\lim_{n\rightarrow\infty} \sqrt[n]{a_lr_l^n}=r_l.$$

    If $|m_i|=|r_l|$ for some $i=1, 2, ...p$, \begin{align*}
        \lim_{n\rightarrow\infty}\sqrt[n]{S_n}&=\lim_{n\rightarrow\infty}\sqrt[n]{(\sum_{i=1}^{l-1} a_ic_i^n + a_l+\cos(n\theta)\sum_{i=1}^{p} b_ie_i^n+i\sin(n\theta)\sum_{i=1}^{p} d_ie_i^n)r_l^n}\\
        &=\lim_{n\rightarrow\infty}\sqrt[n]{(a_l+\cos(n\theta)\sum_{i=1}^{p} b_i+i\sin(n\theta)\sum_{i=1}^{p} d_i)r_l^n}\\
        &=r_l.
    \end{align*}
    If $|m_i|>|r_l|$ for some $i=1, 2, ...p$ (so $|e_i|>1$), keep in mind that 
    $$\lim_{n\rightarrow\infty} \sum_{i=1}^{l-1} a_ic_i^n + a_l+\cos(n\theta)\sum_{i=1}^{p} b_ie_i^n+i\sin(n\theta)\sum_{i=1}^{p} d_ie_i^n\geq 0$$since the average number of ground states per site must be positive and real. Because of this condition, also note that $i\sin(n\theta)\sum_{i=1}^{p} d_ie_i^n=0$. However, if $\cos(n\theta)<0$ for some $n$, this expression becomes negative for some $360x+n$. The only way for $\cos(n\theta)$ to be greater than or equal to 0 for all integers $n$ would be if $\cos(\theta)=1$, which means that the root would be considered a real root $r_k$ where $|r_k|>r_l$.
    
    Then, since the number of ground states per site converges and the total number of ground states must be between 1 and $2^n$, the number of ground states per site must converge to a number between 1 and 2.

    Using similar reasoning, I can also prove the fact for systems $H$ whose characteristic polynomials have repeated roots (with the extra terms with powers of $n$, note that $\lim_{n\rightarrow\infty}(n^x)^{\frac{1}{n}}=1$ by L'Hopital's for some constant $x$).
\end{proof}

\begin{prop}
    The average number of ground states per site for every $k-$Step Fibonacci recursion can be represented by $x$ where $x$ satisfies $x^k-x^{k-1}-x^{k-2}-...-1=0$.
\end{prop}
\begin{proof}
    By the \textbf{Proposition 5.2}, the number of ground states per site must converge to a number between 1 and 2. Therefore, I can set the number as $x$.

    Recall that every $k-$Step Fibonacci recursion satisfies
    \begin{center}
        $S_n=S_{n-1}+S_{n-2}+...+S_{n-k}$.
    \end{center}
    Since $x=\sqrt[n]{S_n}$ when $n\rightarrow \infty$, $x^n=S_n, x^{n-1}=S_{n-1},...x^{n-k}=S_{n-k}$ when $n\rightarrow\infty$. Therefore, I converted the recursion into a polynomial equation, where
    \begin{center}
        $x^n=x^{n-1}+x^{n-2}+...+x^{n-k}$. 
    \end{center}
    Since the number of ground states per site, $x$, is between 1 and 2 and there is clearly only 1 possible positive real root, $x$ must be the positive real root of
    \begin{center}
        $x^k-x^{k-1}-x^{k-2}-...-1=0$.
    \end{center}
\end{proof}
\begin{defin}
    A Pisot number is a number $v$ where $|v|>1$ whose Galois conjugates all lie within the disk, $|z|<1$.
\end{defin}
\begin{defin}
    A Pisot polynomial is an irreducible polynomial of a Pisot number.
\end{defin}

\begin{lemma}
    For every polynomial $x^k-x^{k-1}-x^{k-2}-...-1=0$ where $k\geq 2$, the absolute values of the Galois conjugates of the number of ground states per site are less than 1.
\end{lemma}
\begin{proof}
    This proof is already a well-established result. However, I will include the proof below as a courtesy for the reader \cite{cuculieredoddlordmattics}.

    Because $x^k-x^{k-1}-x^{k-2}-...-1=0=x^k-\frac{x^k-1}{x-1}$ where $k\geq 2$ represents a $k-$Step Fibonacci recursion that has a root between 1 and 2 (1 excluded since $k\geq 2$ but 2 included), multiplying by $x-1$ will not affect the number of roots lying in the disk. Instead, one point at $x=1$ will be added. Therefore, we can look at
    \begin{center}
        $x^{k+1}-2x^k+1=0$
    \end{center}
    and exclude the point at $x=1$.

    Let us set functions $g=x^{k+1}+1$ and $f=2x^k$. Since $|g(x)|<|f(x)|$ when $|x|=m$ for all numbers $m$ slightly greater than 1, $|g(x)|<|f(x)|$ when $m\rightarrow 1$. By Rouché's Theorem (see \cite{Monard2017} for proof), the same number of roots of $g(x)+f(x)$ lie in $|x|\leq1$ as the number of roots of $f(x)$ in $|x|\leq 1$. So, $x^{k+1}-2x^k+1$ has $k$ roots (since $2x^k$ has $k$ roots) in the disk. Suppose that we have some roots $r$ such that $|r|=1$. Then, for every such $r$, it must satisfy the polynomial $r^{k+1}-2r^k+1=0.$ So, $$|r^{k+1}+1|=2|r|^k=2.$$ Because $|r^{k+1}+1|\leq|r|^{k+1}+|1|=2$ by the Triangle Inequality, $$|r^{k+1}+1|=|r^{k+1}|+|1| \textnormal{ and } r^{k+1}=|r^{k+1}|=1.$$ Since $r^{k+1}-2r^k+1=0$, $r^k=1$ too. So, $$r=r^{k+1}/r^k=1.$$ Thus, the only possible value of $r$ when $|r|=1$ is 1. However, since there are no positive real roots other than the one that lies outside $|x|=1$ in the original polynomial, the only point that lies on $|x|=1$ is not a root of $x^n-x^{n-1}-x^{n-2}-...-1=0$. All the other roots of the original polynomial will then lie in the disk $|x|<1$. So, the absolute values of the Galois conjugates of the number of ground states per site for the system are less than 1.
\end{proof}
Since the absolute value of the Pisot root is greater than all the absolute values of its Galois conjugates, I proved the conjecture for all systems that give $k-$Step Fibonacci recursive sequences.

These results can raise some natural questions regarding Pisot numbers and the degrees of degeneracy of more general systems. First of all, can we construct systems that can give all possible Pisot roots between 1 and 2? Conversely, can we use Pisot roots between 1 and 2 to construct all possible systems $H$? If this second statement is true, we will have proved our conjecture.
\section{Conclusion}
Lattice models, with their versatile properties and broad applicability, are an indispensable asset in mathematical physics, contributing uniquely to our understanding of diverse physical systems and their underlying mechanisms.

In this paper, we see examples of quantum lattice models where the degree of degeneracy of the system is a $k-$Step Fibonacci number and the average number of ground states per site is a positive algebraic integer whose Galois conjugates have strictly smaller absolute values. More specifically, I explored a connection between Pisot roots and the average number of ground states per site, raising multiple questions regarding the use of this connection in finding Pisot numbers between 1 and 2 or constructing possible systems of $H$. These questions can be explored in future works.

\section{Acknowledgements}
I would like to thank Kai Xu for his guidance and support throughout the entirety of this project. His insight and mentorship were invaluable through the process, shaping my work and research regarding this topic.


\clearpage
\begin{thebibliography}{100}
\bibitem{AoPS Wiki} AoPS Wiki Administrators. \textit{Characteristic Polynomial}. Art of Problem Solving Online, \url{https://artofproblemsolving.com/wiki/index.php/Characteristic_polynomial}.
\bibitem{cuculieredoddlordmattics} R. Cuculiere, F. Dodd, N. J. Lord, L. E. Mattics (1989). E3008. \textit{The American Mathematical Monthly, 96(2)}, 155–156. \url{https://doi.org/10.2307/2323204}.
\bibitem{Baxter1982} Rodney R.J. Baxter (1982). \textit{Exactly Solved Models in Statistical Mechanics}. Academic Press. \url{https://physics.anu.edu.au/research/ftp/_files/Exactly.pdf}.
\bibitem{Broxson2006-pe} Bobbi Jo. Broxson (2006). \textit{The Kronecker Product}. UNF Digital Commons. \url{https://digitalcommons.unf.edu/etd/25/}.
\bibitem{Clancy} Tyler Clancy. \textit{The Fibonacci Numbers}, 12-14. \url{https://www.whitman.edu/documents/Academics/Mathematics/clancy.pdf}.
\bibitem{luscher1988there} Martin L{\"u}scher and Peter Weisz (1988). \textit{Is There a Strong Interaction Sector in the Standard Lattice Higgs Model?}, 472-478. Physics Letters B. \url{https://www.sciencedirect.com/science/article/pii/B9780444888075501096}.
\bibitem{Selinger_2016} Jonathan V. Selinger (2016). \textit{Introduction to the Theory of Soft Matter: From Ideal Gases to Liquid Crystals}, 7-24. Springer International Publishing. \url{https://books.google.com/books?id=JKlnCgAAQBAJ&printsec=copyright#v=onepage&q&f=false}.
\bibitem{Winters2023} Robert Winters (March 2023). \textit{Linear Algebra - Lecture 6 Notes}. \url{http://math.rwinters.com/E21b/notes/Lecture6.pdf}.
\end{thebibliography}
\end{document}